\documentclass[journal]{IEEEtran}
\usepackage{amsmath,amssymb,times}
\usepackage{multirow,epsfig,cite,psfrag}
\usepackage{times, amsmath}
\usepackage{stfloats}
\usepackage{float}
\usepackage{dsfont}
\usepackage{breqn}
\usepackage{color}
\usepackage{enumerate}
\usepackage{amsfonts}
\usepackage{amsthm}
\usepackage{amsmath}
\usepackage{array}
\usepackage{bbm}
\usepackage{cite}
\usepackage{dsfont}
\usepackage{epsfig}
\usepackage{float}
\usepackage{algorithm}
\usepackage[T1]{fontenc}
\usepackage{graphicx}
\usepackage{indentfirst}
\usepackage{multicol}
\usepackage{subfigure}
\usepackage{url}
\usepackage{multirow}
\usepackage{algpseudocode}
\usepackage{color}
\usepackage{epstopdf}
\usepackage{algpseudocode}
\usepackage[normalem]{ulem}
\usepackage{float}
\usepackage[nolist]{acronym}
\usepackage{booktabs}
\usepackage{subfigure}
\usepackage{graphicx}

\usepackage{graphicx}
\usepackage{caption}
\usepackage{algorithm}
\usepackage{algpseudocode}

\newtheorem{proposition}{\bf Proposition}

\begin{document}
\title{ A Directed Information Learning Framework for Event-Driven M2M Traffic Prediction}
	\author{\IEEEauthorblockN{Samad Ali \thanks{This research was supported by the by Academy of Finland's Flagship Programme project 6Genesis under Grant 318927 and, in part, by the Office of Naval Research (ONR) under Grant N00014-15-1-2709 and, in part, by the U.S. National Science Foundation under Grant CNS-1739642.}, Walid Saad, and Nandana Rajatheva}
\vspace{-1cm}
	}
\maketitle
\begin{abstract}
Burst of transmissions stemming from event-driven traffic in machine type communication (MTC) can lead to congestion of random access resources, packet collisions, and long delays. In this paper, a directed information (DI) learning framework is proposed to predict the source traffic in event-driven MTC. By capturing the history of transmissions during past events by a sequence of binary random variables, the DI between different machine type devices (MTDs) is calculated and used for predicting the set of possible MTDs that are likely to report an event. Analytical and simulation results show that the proposed DI learning method can reveal the correlation between transmission from different MTDs that report the same event, and the order in which they transmit their data. The proposed algorithm and the presented results show that the DI can be used to implement effective predictive resource allocation for event-driven MTC.
\end{abstract}
\vspace{-0.25cm}
\section{Introduction} \label{sec:introduction}
Next-generation wireless networks must be able to support massive machine-type-communications (MTCs) \cite{dawy}. Due to their short-packet nature, MTCs must be supported by cellular networks with minimal signaling overhead. In fact, the radio resources required for the cellular random access (RA) process in conventional systems is higher than the total amount of resources required for data transmission in most IoT applications \cite{RACHM2M2}. Moreover, since machine-type-devices (MTD) in the IoT randomly select radio resources from a pool of RA slots, if the same RA resource is selected by more than one MTD, collisions occur and the RA process fails. Such collisions become a very challenging problem particularly in massive access scenarios \cite{MassiveM2M}. To overcome some of these challenges, there has been an increased recent interest in the notion of a fast uplink grant, a method proposed by 3GPP to grant fast uplink access to MTDs without requiring them to perform uplink scheduling requests \cite{ali2018fast}. This solution can potentially solve the problems that conventional RA face in MTC such as collisions, congestion, and inefficiency due to large signaling overhead compared to the actual data size. Naturally, MTD selection for fast uplink grant and uplink radio resource allocation are performed at the base station (BS).

Most of the research in MTC is focused on increasing the efficiency of the RA process, or back-off mechanisms to reduce the congestion at the cost of long delays \cite{RACHM2M2}. However, virtually no prior work analyzed how MTDs can be selected for fast uplink grant. Moreover, uplink grant allocation requires sophisticated source traffic prediction mechanisms, that is, to predict the set of active MTDs at each time. Without proper predictions, uplink grant allocation can lead to a waste of resources \cite{ali2018fast} and \cite{LTE14outlook}. Prediction of periodic and semi-periodic traffic can be done by using calendar-based pattern mining tools \cite{periodicpattern}. In event-driven MTC, it is impossible to predict the unexpected events and their associated MTD transmissions. Prior work on traffic modeling for MTC is moslty focused on the aggregate traffic model at BS-level. The statistical properties of bursty traffic are typically modeled using a beta distribution (e.g, see\cite{3GPPMTC_betta} and \cite{m2mtrafficstudy}). In \cite{MTCsourceTraffic}, the source traffic modeling at individual MTD level is studied and modeled using a Markvo modulated Poisson process (MMPP). However, if there is a correlation between the MTD transmissions related to some IoT event, that correlation can be exploited. The authors in \cite{kalor2018random} propose to exploit such a correlation for optimizing the RA in MTC. In contrast, here, we propose to exploit the correlation for source traffic prediction during IoT events. In other words, upon detection of a certain IoT event, the network can predict which other MTDs will experience the same event and are likely to start transmitting. Such a prediction will enable implementation of the fast uplink grant for event-driven transmissions. This will also help alleviate event-driven transmission problems such as congestion of RA resources, collisions, and long delay. Clearly, event-driven source traffic prediction is a key component of the general problem of optimal fast uplink grant allocation.

The main contribution of this paper is to introduce a novel learning method to find the set of MTDs that face the same event once an event happens, by using data from previous events. The concept of \emph{directed information (DI)} \cite{massey1990causality} is introduced to study the statistical causality between the transmission patterns of different MTDs. DI is a powerful tool to investigate the causality and flow of information between sequences of random variables. DI can capture how an IoT event will propagate and, hence, it is apropos for MTC source traffic prediction. Our results show how the amount of correlation between transmission pattern of different MTDs in IoT events can be inferred by using DI. Moreover, our results show how the value of DI can be used to find the order of transmissions of MTDs related to certain IoT events. 
\vspace{-0.3cm}
\section{System model and proposed framework}\label{sysmodel}
Consider the uplink of a cellular system composed of one BS and a set $\mathcal{M}$ of $M$ MTDs that are considered to be static (or having slow mobility). First, we consider that MTDs have been transmitting for a long period of time in a conventional manner, i.e., sending scheduling requests and transmitting data packets after receiving uplink grants. We assume that the BS has gathered the history of transmissions from all MTDs in the network. For periodic transmissions, the BS can use this previously gathered data to learn the time instants during which the MTDs will have data to transmit using calendar-based pattern mining methods \cite{periodicpattern}. Hence, after learning the periodic transmission patterns, the MTDs no longer perform RA and, instead, a fast uplink grant system is implemented for allocating uplink resources. However, for unexpected events \cite{park2016learning} that occur beyond the periodic patterns, the MTDs will have to send scheduling requests. Consider a set $\mathcal{K} \subset \mathcal{M}$ of $K$ MTDs that have new data packets to transmit in order to report an unexpected event. This event will, hence, lead to a burst of scheduling requests from the MTDs in set $\mathcal{K}$. This burst of scheduling requests will lead to RA channel congestion, long delays, and eventually, waste of resources. Our objective is to detect such events when they happen, and predict the set $\mathcal{K}$ for each event. Once the set $\mathcal{K}$ is predicted upon event detection, the network can allocate uplink resources for as many of MTDs in set $\mathcal{K}$ as possible before they start sending scheduling requests. Such predictions help avoid RA congestion, reduce delays, and improve fast uplink grant efficiency.

Since the MTDs will not send scheduling requests for periodic transmissions, any scheduling request from MTDs is seen as an \emph{event trigger} at the BS. Naturally, once an event happens, some MTDs will initiate the access process before others. This is a practical assumption since most IoT events propagate through a geographical area, and sensors sense them at different times \cite{park2016learning}. As discussed earlier, the BS will have information on a history of prior data traffic for event-driven MTCs. This data along with event-triggered MTDs can be used to predict the set $\mathcal{K}$. Such a prediction problem can be modeled with the paradigm of \emph{causality}. The fundamental question here is: \emph{By looking at the transmission patterns of MTDs during past events, once a specific MTD starts transmission, can the network determine which other MTDs will start transmitting?} To answer this question, one must look at the \emph{causal relationship} between transmission patterns of MTDs in past events. We propose a novel approach based on the concept of \emph{DI} introduced in \cite{massey1990causality} to infer the causal relationship between different MTDs and, hence, predict the set of nodes with data to transmit. The history of RA requests from each MTD can be represented by a time series. Then, DI will be used to infer the causality between different MTDs and predict the nodes that face the same event.
\vspace{-0.3cm}
\section{Source Traffic Prediction Using Directed Information}
Without loss of generality, we present our method and algorithm for two MTDs. This can be naturally extended to any number of MTDs. Consider two MTDs, indexed by $X$ and $Y$. Here, ${X}^N = \{X_1,X_2,...,X_N\}$ is defined as a length $N$ sequence of random variables for MTD $X$ and, similarly, ${Y}^N = \{Y_1,Y_2,...,Y_N\}$ is defined for MTD $Y$. $X_i$ represents element $i$ of ${X}^N$, and ${X}_l^N = \{X_l,X_{l+1},...,X_N\}$. The notion of DI introduced in \cite{massey1990causality} is a measure of information flow between two sequences. The DI from sequence ${X}^N$ to sequence ${Y}^N$ is denoted by $I({X}^N \rightarrow {Y}^N)$ and is defined as:
\begin{equation}\label{directedInformation1}
I({X}^N \rightarrow {Y}^N) = H({Y}^N) - H({Y}^N\|{X}^N),
\end{equation}
where $H({Y}^N)$ is the entropy of the $N$ dimensional random sequence ${Y}^N$ and $H({Y}^N\|{X}^N)$ is the entropy of ${Y}^N$ causally conditioned on ${X}^N$, which is defined as:
\begin{equation}\label{causuallyconditioned}
H({Y}^N\|{X}^N) = \sum_{i=1}^{N} H(X^i;Y_i|Y^{i-1}).
\end{equation}
Combining (\ref{directedInformation1}) and (\ref{causuallyconditioned}) yields the DI:
\begin{equation}\label{DIfinal}
I({X}^N \rightarrow {Y}^N) = \sum_{i=1}^{N} I(X^i;Y_i|Y^{i-1}),
\end{equation}
where $I(X^i;Y_i|Y^{i-1})$ is the mutual information between $X^i$ and $Y_i$ conditioned on $Y^{i-1}$. The DI provides an interpretation of how much $Y$ is statistically causally influenced by $X$. By calculating the DI between any pairs of sequences of scheduling requests of MTDs, it is possible to infer whether an MTD causally affects another MTD or not. In MTC, one MTD's reporting will not necessarily cause other MTDs to report, however, since MTDs often monitor a related environment, events can be correlated. Indeed, by observing the sequence of transmissions, one can infer a \emph{statistical causality} between the two time series of transmission sequences. Hence, it is possible to infer the direction of event propagation between MTDs. DI is a metric that can capture this causality. For example, DI is used in \cite{behnamDI} and  \cite{nimasoltani} to predict seizures in epilepsy patients using electrocardiography recordings of the brain and to infer causality between neurons. To understand the dynamics of DI, consider an example of a connection between two time series $X^N$ and $Y^N$ that are related as follows:

\begin{equation}\label{yConstruct}
Y_n = \sum_{i=0}^{N-1} \beta_{n-i} X_{n-i} + \alpha_{n-i-1} Y_{n-i-1}+ N_n,
\end{equation}
where $\beta_{n-i}$ and $\alpha_{n-i-1}$ are the factors that show the relationship between the current value of $Y$ with the previous values of $Y$, the previous and current value of $X$. (\ref{yConstruct}) shows the relation between $Y_n$ and $Y_{n-1}$ and all the $N$ previous elements of $X$. If $\beta_n\neq0$, then, there is a causal connection from $X$ to $Y$ and vice versa. In general, we can argue that for any $\beta_{n-i} \neq 0$, there is causal connection from $X$ to $Y$. But if $\beta_n = 0$, then $I({X}^N \rightarrow {Y}^N) \neq 0$ but $I({Y}^N \rightarrow {X}^N)= 0$ which implies that $Y$ does not cause $X$.
\subsection{Distributions of Random Access}
For a fixed event length $N$, each MTD might perform several scheduling requests at some of the time steps $k \in \{1,2,...,N\}$ and remain silent in other times. Binary variables are used to represent the elements of this sequence of scheduling requests. If a scheduling request was sent at time $k$, then $X_k=1$ and $X_i =0$ otherwise. We assume that several sequences from different events are available for each MTD. These realizations of each sequence should be used to estimate the probability distributions needed to calculate the DI in (\ref{DIfinal}) between different MTDs. To calculate this, the probability distribution of each term on the right-hand side of (\ref{DIforpairs}) is required. To this end, first, for each MTD, the probability of having data to transmit or not at each time $i$ of the event is modeled using a Bernoulli distribution with parameter $p_i$. $p_i$ is the probability of having data to transmit $(P(X=1) = p_i)$ and $1-p_i$ is the probability of no data to transmit $(P(X=0) = 1-p_i)$. To calculate the DI between $X$ and $Y$, we need to estimate the joint probability distributions and also multivariate probability distribution between events of MTDs. To do this, we use the multivariate Bernoulli distribution that plays a fundamental role in calculating DI in (\ref{DIfinal}). Probibility of any sequence $X = (X_1,...,X_L)$ in multivariate Bernoulli distribution is given by:
\begin{equation}\label{multivariatebernoulli}
P(x_1,...,x_l)\!=\!P(X_1=x_1,...,X_L=x_l),
\end{equation}
where $x_i\in\{0,1\}$. Such distributions are estimated from the gathered data and then used to calculate the DI. However, the exact calculation of the DI (\ref{DIfinal}) in practice is very hard because it requires finding all of the probability distributions given in (\ref{DIfinal}). Because of this, approximations of entropies and the DI are often performed in the literature \cite{UniverstalEstimation}. In our algorithm, instead of approximations, we use a more precise method, that relies on calculating the DI between sequences of length two. For such sequences,  $\{X_k X_{k+1},Y_kY_{k+1}\}$, the directed DI has the following expression:
\begin{align}\label{DIforpairs}
I({X}_{k}^{k+1} \rightarrow {Y}_{k}^{k+1} ) =  & I(X_k;Y_k) + I(X_k X_{k+1};Y_{k+1}|Y_k) \nonumber \\
= & H(X_k)- H(X_k Y_{k})+H(X_k X_{k+1} Y_k)\nonumber\\
 &+\!H(Y_k Y_{k+1})-\!H(X_k X_{k+1} Y_k Y_{k+1}).
\end{align}
\vspace{-1cm}
\subsection{Bounds on Directed Information and Causality}
\begin{proposition}\label{lemmadi}
	For a multivariate binary distribution, the DI in (\ref{DIforpairs}) lies in the range $[0, \quad 2]$.
\end{proposition}
\begin{proof}
	We first use (\ref{DIforpairs}) to evaluate the causality between two times series $X$ and $Y$. If MTD $Y$ always starts its RA process after MTD $X$, i.e., $Y$ is deterministic given $X$, then equation (\ref{DIforpairs}) simplifies to:
	\begin{align}\label{determenistic}
	I({X}_{k}^{k+1} \rightarrow {Y}_{k}^{k+1} ) =  & I(X_k;Y_k) + I(X_k X_{k+1};Y_{k+1}|Y_k) \nonumber \\
	= & H(Y_k Y_{k+1}) - H(Y_k|X_k X_{k+1})\nonumber\\
	&- H(Y_k|X_k)-H(Y_k Y_{k+1}|X_k X_{k+1}) \nonumber\\
	= & H(Y_k Y_{k+1}),
	\end{align}
	since $H(Y_k|X_k X_{k+1}) = 0$ and $H(Y_k Y_{k+1}|X_k X_{k+1}) = 0$ due to the fact that there is no uncertainty left in $Y$ given $X$. The maximum value for $H(Y_k Y_{k+1})$ happens when all the four possible outcomes have equal value of $\frac{1}{4}$, which results in $H(Y_k Y_{k+1}) = 2$. Moreover, if $X$ gives no information about $Y$, meaning that $X$ always happens after $Y$, then:
	\begin{align}\label{noinformation}
	I({X}_{k}^{k+1} \!\rightarrow\!{Y}_{k}^{k+1} )=& I(X_k;Y_k) + I(X_k X_{k+1};Y_{k+1}|Y_k) \nonumber \\
	= & H(Y_k Y_{k+1}) - H(Y_k|X_k X_{k+1}) \nonumber\\
	&- H(Y_k|X_k)-H(Y_k Y_{k+1}|X_k X_{k+1}) \nonumber\\
	=& H(Y_k\!Y_{k+1})\!-\!H(Y_k)\!-\!H(Y_k)\!-\!H(Y_k\!Y_{k+1}) \nonumber\\
    =& 0,
	\end{align}
	since giving $X$ goes not provide any extra information about $Y$ and hence, $H(Y_k|X_k X_{k+1}) = H(Y_k)$ and $H(Y_k Y_{k+1}|X_k X_{k+1}) = H(Y_k Y_{k+1})$. Hence, if the DI between two pairs are zero, then there is no causal connectivity from $X$ to $Y$.
\end{proof}
Proposition \ref{lemmadi} quantifies the maximum amount of flow of the information between any pair of sequences and it can be used to calculate the error rate of the prediction algorithm. The maximum DI means $100\%$ accuracy in the prediction of the learning algorithm.
\begin{algorithm}[t!] \footnotesize
	\caption{Directed Information Based Traffic Prediction}
\begin{algorithmic}[1]
	\State Fix the length of events to the maximum length $L$ amoung events.
	\State For every time step that event had happened, set the value to $1$ and $0$ otherwise.
	\State \textbf{for} every pair of MTDs $X$ and $Y$ \textbf{do:}
	\State \quad \textbf{for} every pair $X_kX_{k+1}$ and $Y_kY_{k+1}$, $k\in\{1,...,L-1\}$ and $X_{k}X_{k+1}$ and $Y_{k+i}Y_{k+i+1}$, $i \in \{1,..,L-k\}$ \textbf{do:}
	\State \quad Calculate the probability distributions (\ref{multivariatebernoulli}) for entropies (\ref{DIforpairs}).
	\State \quad Calculate DI between from MTD $X$ to MTD $Y$ and vice versa (\ref{DIforpairs}).
	\State \quad Create the set $\mathcal{K}$ for every MTD, such that $I(X_kX_{k+1} \rightarrow  Y_{k+i}Y_{k+i+1}) \neq 0$, and include $i$ and DI.
	\State \quad \textbf{end}
	\State  \textbf{end}
\end{algorithmic}
\label{dialgorithm}
\end{algorithm}
\vspace{-0.5cm}
\subsection{Proposed Directed Information Algorithm}
To develop our learning algorithm, we use (\ref{DIforpairs}) to calculate the DI between all combinations of pairs of event sequences. For every $k\in\{1,..,N-1\}$, and $i\in\{1,...,I: I<N-k\}$, the DI between $X_kX_{k+1}$ and $Y_{k+i}Y_{k+i+1}$ is calculated. To do this, first, the multivariate probability distributions in (\ref{multivariatebernoulli}) are estimated for all of the entropy terms on the right hand side of (\ref{DIforpairs}). Then, entropies are calculated and finally the DI is derived from (\ref{DIforpairs}). For every non-zero DI, there is causal connectivity between $X_kX_{k+1}$ and $Y_{k+i}Y_{k+i+1}$. This means that once an event happens at MTD $X$ during time $k$ to $k+1$, after $i$ time steps, MTD $Y$ will face a similar event. Moreover, using Proposition \ref{lemmadi}, when the DI is at its maximum, the prediction error of the algorithm is zero. In other words, the algorithm can surely predict that MTD $Y$ will transmit $i$ time steps after MTD $X$. For values between the minimum and maximum, a higher DI value implies a higher probability of transmission, and, hence, the BS can allocate grants to MTDs having a larger activation probability. So, once an event happens at one MTD, for every other MTD, we can predict the time instance at which the other MTD is likely to transmit. To fully implement this method, for every MTD $X$, we can create a set $\mathcal{K}$, that consists of all the other MTDs $Y$ such that $I(X \rightarrow Y) \neq 0$, and a time value $i$ that shows how many steps into the future, MTD $Y$ will start transmitting after MTD $X$. This set also includes $I(X_kX_{k+1} \rightarrow  Y_{k+i}Y_{k+i+1}) \neq 0$, which captures the value of the DI, pertaining to how likely it is for $Y$ to transmit $i$ time steps after $X$. This approach is summarized in Algorithm \ref{dialgorithm}. 

In terms of complexity and scalability, the proposed learning process, including probability modeling, and DI calculations can be done offline at a cloud server. Moreover, the set of IoT devices that face a similar event in real-wolrd IoT networks is often smaller than the total number of devices in the network, thus further reducing the complexity of DI computation in real-world networks. Let us consider that the events have length $L$, then for any DI calculation, the network must compute $5L(L-1)^2$ values of entropy from the dataset. Hence, the algorithm will have a complexity of $\mathcal{O}(L^3)$, which is polynomial and reasonable. Moreover, for a network with $n$ devices, the total number of DI pairs that must be calculated is ${n\choose2} = n(n-1)/2$, and hence, the complexity in term of number of devices is also polynomial, i.e., $\mathcal{O}(n^2)$.

\section{Simulation Results}
We first generate artificial data to represent the transmission patterns of the MTDs. Four MTDs are considered, $X$, $Y$, $Z$, and $T$. The length of the IoT event is considered to be 12 time steps. For MTD $X$, the event happens at times $t\in \{1,2,3,4, 7,8,9\}$. In other words, at those times, there can be a transmission, and at other times, MTD $X$ is silent. An example transmission pattern of MTD $X$ can be $X^{12} = \{101100111000\}$. For MTD $Y$, we consider $t\in \{4,5,6, 8,9,10,11\}$. For MTD $Z$, we assume that the event occurs $3$ time steps after $X$ with probability $0.8$ and, for MTD $T$, the event happens $2$ time steps after $X$ with probability  $0.2$. The results are presented in Fig. \ref{result1}. In our results, $I({X_{k,k+1}} \rightarrow {Y_{k+i-1,k+i}})$ is presented in element $(i,k)$ of the matrix of the results $I({X} \rightarrow {Y})$. Given that in our scenario the event at MTD $X$ often happens before MTD $Y$, we obtain higher $I({X} \rightarrow {Y})$ compared to $I({Y} \rightarrow {X})$, as shown in Fig. \ref{result1}. However, since, $Y$ experiences the event at times $\{4,5,6\}$ and $X$ at times $\{7,8,9\}$, there is some flow of information between $Y_4^9$ and $X_7^9$. For $X$, $Z$, and $T$, since $Z$ happens after $X$ with a high probability, i.e.,  $I({X} \rightarrow {Z})$ is higher than $I({X} \rightarrow {T})$. However, since some events in $Z$ happen before $X$, there is some small flow of information from $Z_4^7$ to $X_7^9$. A similar argument holds for $I({T} \rightarrow {X})$. So, for example, in an IoT scenario, if MTD $X$ starts transmitting, we can say that with a very high probability, MTD $Y$ will transmit $4 , 7$ and $10$ time steps later. There is a high probability that $Z$ will transmit $4$ time steps later, and the probability of $T$ transmitting after $X$ is lower. If $X$ keeps transmitting and transmits at time step $3$, by looking at element $(2,3)$ of the results in Fig. \ref{result1}, we can argue that after $1$ time step the probability of $Y$ and $Z$ transmitting is high, and $T$ has a low probability for transmission. The presented results show the power of DI in inferring causality between such sequences of data, and how it can be used to infer causal relationships between the transmission patterns of different MTDs.
\begin{figure}[t!]
	\centering
	\includegraphics[width=8.5cm]{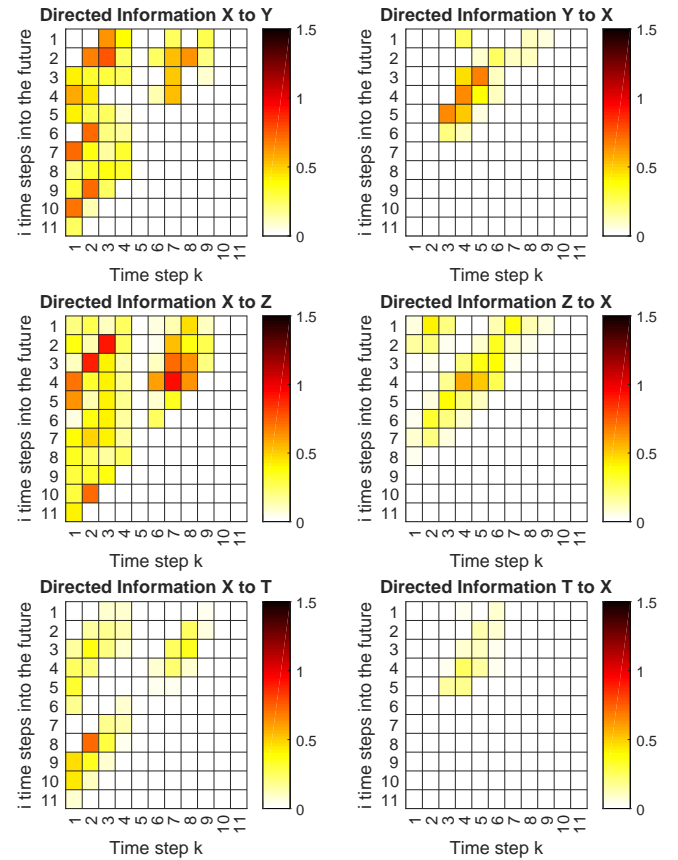}
	\caption{DI between pairs of IoT devices. $I({X_{k,k+1}} \rightarrow {Y_{k+i-1,k+i}})$ is presented in element $(i,k)$ of $I(X \rightarrow Y)$.}\label{result1}
\end{figure}
\section{Conclusion}
In this paper, we have presented a novel approach to predict the source traffic in event-driven MTC. First, we have explored how MTD transmissions can be correlated due to certain relationships between IoT events. Then, we have used a binary sequence to model the transmission history of different MTDs. We have introduced the concept of DI, and shown how it can infer causality between different sequences. Finally, we have developed an algorithm, that can predict the set of MTDs that have data to transmit. Simulation results have shown the dynamics and behavior of DI. These results present how DI can be used to infer causality between transmission patterns of MTDs, and, how this inference can be used to predict the source traffic in the MTC networks.
\vspace{-0.2cm}
\bibliographystyle{IEEEtran}
\bibliography{di-refs}

\begin{thebibliography}{10}
\providecommand{\url}[1]{#1}
\csname url@samestyle\endcsname
\providecommand{\newblock}{\relax}
\providecommand{\bibinfo}[2]{#2}
\providecommand{\BIBentrySTDinterwordspacing}{\spaceskip=0pt\relax}
\providecommand{\BIBentryALTinterwordstretchfactor}{4}
\providecommand{\BIBentryALTinterwordspacing}{\spaceskip=\fontdimen2\font plus
\BIBentryALTinterwordstretchfactor\fontdimen3\font minus
  \fontdimen4\font\relax}
\providecommand{\BIBforeignlanguage}[2]{{%
\expandafter\ifx\csname l@#1\endcsname\relax
\typeout{** WARNING: IEEEtran.bst: No hyphenation pattern has been}%
\typeout{** loaded for the language `#1'. Using the pattern for}%
\typeout{** the default language instead.}%
\else
\language=\csname l@#1\endcsname
\fi
#2}}
\providecommand{\BIBdecl}{\relax}
\BIBdecl

\bibitem{dawy}
Z.~Dawy, W.~Saad, A.~Ghosh, J.~G. Andrews, and E.~Yaacoub, ``Toward massive
  machine type cellular communications,'' \emph{IEEE Wireless Communications},
  vol.~24, no.~1, pp. 120--128, February 2017.

\bibitem{RACHM2M2}
A.~Laya, L.~Alonso, and J.~Alonso-Zarate, ``Is the random access channel of lte
  and lte-a suitable for m2m communications? a survey of alternatives.''
  \emph{IEEE Communications Surveys and Tutorials}, vol.~16, no.~1, pp. 4--16,
  2014.

\bibitem{MassiveM2M}
C.~Bockelmann, N.~Pratas, H.~Nikopour, K.~Au, T.~Svensson, C.~Stefanovic,
  P.~Popovski, and A.~Dekorsy, ``Massive machine-type communications in 5g:
  physical and mac-layer solutions,'' \emph{IEEE Communications Magazine},
  vol.~54, no.~9, pp. 59--65, September 2016.

\bibitem{ali2018fast}
S.~Ali, N.~Rajatheva, and W.~Saad, ``Fast uplink grant for machine type
  communications: Challenges and opportunities,'' \emph{arXiv preprint
  arXiv:1801.04953}, 2018.

\bibitem{LTE14outlook}
C.~Hoymann, D.~Astely, M.~Stattin, G.~Wikstrom, J.-F. Cheng, A.~Hoglund,
  M.~Frenne, R.~Blasco, J.~Huschke, and F.~Gunnarsson, ``{LTE} release 14
  outlook,'' \emph{IEEE Communications Magazine}, vol.~54, no.~6, pp. 44--49,
  2016.

\bibitem{periodicpattern}
J.~Adhikari and P.~Rao, ``Identifying calendar-based periodic patterns,''
  \emph{Emerging paradigms in machine learning}, pp. 329--357, 2013.

\bibitem{3GPPMTC_betta}
3GPP, ``Study on {RAN} improvements for machine type communications,'' {3rd
  Generation Partnership Project (3GPP)}, Technical Specification (TS) 37.868,
  September 2011.

\bibitem{m2mtrafficstudy}
M.~Centenaro and L.~Vangelista, ``A study on m2m traffic and its impact on
  cellular networks,'' in \emph{2015 IEEE 2nd World Forum on Internet of Things
  (WF-IoT)}, Dec 2015, pp. 154--159.

\bibitem{MTCsourceTraffic}
M.~Laner, P.~Svoboda, N.~Nikaein, and M.~Rupp, ``Traffic models for machine
  type communications,'' in \emph{ISWCS 2013; The Tenth International Symposium
  on Wireless Communication Systems}, Aug 2013, pp. 1--5.

\bibitem{kalor2018random}
A.~E. Kal{\o}r, O.~A. Hanna, and P.~Popovski, ``Random access schemes in
  wireless systems with correlated user activity,'' \emph{arXiv preprint
  arXiv:1803.03610}, 2018.

\bibitem{massey1990causality}
J.~Massey, ``Causality, feedback and directed information,'' in \emph{Proc.
  Int. Symp. Inf. Theory Applic.(ISITA-90)}, 1990, pp. 303--305.

\bibitem{park2016learning}
T.~Park, N.~Abuzainab, and W.~Saad, ``Learning how to communicate in the
  internet of things: Finite resources and heterogeneity,'' \emph{IEEE Access},
  vol.~4, pp. 7063--7073, 2016.

\bibitem{behnamDI}
R.~Malladi, G.~P. Kalamangalam, N.~Tandon, and B.~Aazhang, ``Inferring causal
  connectivity in epileptogenic zone using directed information,'' in
  \emph{Proc. of IEEE International Conference on Acoustics, Speech and Signal
  Processing (ICASSP)}, Brisbane, Australia, April 2015, pp. 822--826.

\bibitem{nimasoltani}
N.~Soltani and A.~J. Goldsmith, ``Directed information between connected leaky
  integrate-and-fire neurons,'' \emph{IEEE Transactions on Information Theory},
  vol.~63, no.~9, pp. 5954--5967, Sept 2017.

\bibitem{UniverstalEstimation}
J.~Jiao, H.~H. Permuter, L.~Zhao, Y.-H. Kim, and T.~Weissman, ``Universal
  estimation of directed information,'' \emph{IEEE Transactions on Information
  Theory}, vol.~59, no.~10, pp. 6220--6242, 2013.

\end{thebibliography}
\end{document}